\documentclass[11pt, a4paper]{article}
\usepackage{times}
\usepackage{graphicx}
\usepackage{latexsym}
\usepackage[margin=2.5cm]{geometry}

\usepackage{amsmath,amssymb,amsthm}
\usepackage[mathic]{mathtools}
\usepackage[T1]{fontenc}
\usepackage{xcolor}
\usepackage{dsfont}
\usepackage{microtype}
\usepackage{environ}
\usepackage{enumitem}
\usepackage{bm}
\usepackage{comment}
\usepackage{tikz}
\usetikzlibrary{arrows,automata, shapes.geometric, positioning}
\usepackage{xspace}
\usepackage{graphicx}
\usepackage[ruled, vlined, longend, linesnumbered]{algorithm2e}
\DontPrintSemicolon
\SetKwInOut{Input}{Input}\SetKwInOut{Output}{Output}

\newtheorem{thrm}{Theorem}
\newtheorem{observation}[thrm]{Observation}
\newtheorem{lemma}[thrm]{Lemma}

\newtheorem{fact}[thrm]{Fact}
\newtheorem{definition}[thrm]{Definition}

\DeclareFontShape{OT1}{cmr}{bx}{sc}{<-> cmbcsc10}{}
\usepackage[]{hyperref}

\NewEnviron{cproblem}[1]{%
\begin{center}\fbox{\parbox{4in}{%
    {\centering\scshape #1\par}%
    \parskip=1ex
    \everypar{\hangindent=1em}%
    \BODY
}}\end{center}}

\newcommand{\ones}{\mathds{1}}

\DeclareMathOperator{\sh}{Sh}
\newcommand{\shap}{\phi^{\sh}}
\newcommand{\Sshap}{S^*}

\newcommand{\nb}[2]{\#_{#1}(#2)}

\newcommand{\MGSS}{\textsc{Max-Shapley-Group}\xspace }
\newcommand{\HMHS}{\textsc{Harmonic-Max-Hitting-Set}\xspace }

\newcommand{\MHS}{\textsc{Max-Hitting-Set}\xspace }

\newcommand{\RR}{\ensuremath{\mathbb{R}}}

\newcommand{\RRR}{\mathcal{R}}
\newcommand{\PPP}{\mathcal{P}}

\DeclareMathOperator{\E}{E}
\newcommand{\GSC}{\hat \phi^{\sh}}
\newcommand{\tGSC}{\tilde \phi^{\sh}}

\newcommand{\DKS}{\textsc{Densest-\(k\)-Subgraph}\xspace }

\begin{document}

\title{Maximizing Influence-based Group Shapley Centrality}

\author{Ruben Becker \and Gianlorenzo D'Angelo \and Hugo Gilbert}
\date{Gran Sasso Science Institute, L'Aquila, Italy\\firstname.lastname@gssi.it} 

\maketitle

\begin{abstract}
    One key problem in network analysis is the so-called influence maximization problem, which consists in finding a set $S$ of at most $k$ seed users, in a social network, maximizing the spread of information from $S$. This paper studies a related but slightly different problem: We want to find a set $S$ of at most $k$ seed users that maximizes the spread of information, when $S$ is added to an already pre-existing -- \emph{but unknown} -- set of seed users $T$. We consider such scenario to be very realistic. Assume a central entity wants to spread a piece of news, while having a budget to influence $k$ users. This central authority may know that some users are already aware of the information and are going to spread it anyhow. The identity of these users being however completely unknown.
    
    We model this optimization problem using the Group Shapley value, a well-founded concept from cooperative game theory. While the standard influence maximization problem is easy to approximate within a factor $1-1/e-\epsilon$ for any $\epsilon>0$, assuming common computational complexity conjectures, we obtain strong hardness of approximation results for the problem at hand in this paper. Maybe most prominently, we show that it cannot be approximated within $1/n^{o(1)}$ under the Gap Exponential Time Hypothesis. Hence, it is unlikely to achieve anything better than a polynomial factor approximation. Nevertheless, we show that a greedy algorithm can achieve a factor of $\frac{1-1/e}{k}-\epsilon$ for any $\epsilon>0$, showing that not all is lost in settings where $k$ is bounded. 
\end{abstract}

\section{INTRODUCTION}
Node centrality and propagation of information or influence are two main topics in network analysis. The former regards the problem of determining the most important nodes in a network according to some measure of importance, while the latter studies mathematical models to represent how information propagates in a communication network or how the influence of individuals spreads in a network.

In order to measure the centrality of nodes in a network a real-valued function, called centrality index, associates a real number with each node that reflects its importance or criticality within the network. 
Most of the centrality indices defined in the literature are based on graph-theoretical concepts and static graph properties like distance (closeness, harmonic, and degree), spectral (page-rank or Katz), or path-based (betweenness, coverage) properties. 
Modeling the spread of influence, instead, requires the combination of a \emph{dynamic} model for influence diffusion and a static model based on the network topology.

Chen and Teng~\cite{chen2017interplay} initiated the study of the interplay between spreading dynamics and network centrality by defining two centrality indices based on dynamic models for influence diffusion: the \emph{single node influence centrality}, which measures the centrality of a node by its capability of spreading influence when acting alone, and the \emph{Shapley centrality}, which uses the Shapley value to measure the capability of a node to increase the spreading capacity of a group of nodes.

In cooperative game theory, the Shapley value assesses the expected relevance of each player within a subset of players (also called coalition), where the expectation is taken over all possible coalitions. More formally, given a characteristic function $\tau$ that maps each coalition to the total payoff that this coalition receives, the Shapley value of a player $i$ can be understood as the expected payoff that $i$ adds to any coalition, w.r.t. function $\tau$. 
The Shapley centrality index studied by Chen and Teng~\cite{chen2017interplay} measures the centrality of a node by using the Shapley value and the spreading function $\sigma$ as characteristic function.

Most centrality indices neglect the relevance that coalitions of individuals and their coordination play in social networks. For this reason, many centrality indices have been generalized to \emph{group centrality indices} which are real-valued functions over subsets of nodes instead of single nodes. Typically, a group centrality index is fundamentally different from a combination of the individual centrality indices of the nodes that compose the group, as it captures the relevance of the set as a whole, and not just as a sum of individuals.

This paper extends the notion of influence-based Shapley centrality from single nodes to \emph{groups of nodes} by using the concept of the Group Shapley value.  
Our Influence-based Group Shapley (IGS) centrality associates to a set $S$ of nodes, the expected gain in influence that $S$ adds to any pre-existing seed set $T$.  
Notably, we investigate the problem of finding a set \(S\) of size at most $k$ with highest IGS value.  

Interestingly, we believe that this way of evaluating the importance of a set of seed users is of high interest from a practical viewpoint. Assume a central entity wants to spread a given piece of news, while having a budget to influence a set of $k$ users, at the same time knowing that already some users are aware of the information and are going to spread it anyhow. The central entity, however, may have no knowledge about who these users are. In this case, the central authority should target a set of seed users with large IGS value.

\paragraph{Our contribution.} We formalize the \MGSS problem  of finding a set of seed nodes with highest IGS centrality under a cardinality constraint and show how to compute a $(1-\epsilon)$-approximate value for the IGS centrality of a given set of nodes. Unfortunately, assuming common complexity theory conjectures, we obtain strong hardness of approximation results for the \MGSS problem. Maybe most prominently, we show that it cannot be approximated within $1/n^{o(1)}$ under the Gap Exponential Time Hypothesis. Hence, it is unlikely to achieve an approximation factor that is better than a polynomial in $n$. Nevertheless, we show that a greedy algorithm achieves a factor of $\frac{1-1/e}{k}-\epsilon$ for any $\epsilon>0$, showing that not all is lost in settings where $k$ is bounded.

\section{RELATED WORK}
There is a large literature about network centrality indices, see~\cite[Ch.~7]{N10} for an introduction. Centrality indices are usually categorized as distance-based (e.g., closeness centrality~\cite{B50}), path-based (e.g., betweenness centrality~\cite{F77}), or spectral (e.g., page-rank~\cite{BP98}). Most of the literature focuses on defining indices for specific application domains (e.g.,~\cite{B50,BP98,F77}), on the efficient computation of the centrality index of each node or of the top-ranked nodes (e.g.,~\cite{BergaminiBCMM19,RiondatoU18}), or on axiomatic characterization (e.g.,~\cite{AltmanT05,BV14}). 
Several centrality indices have been generalized to \emph{group} centrality indices~\cite{AngrimanGBZGM19,BergaminiGM18,ChenWW16,EveretB99,IshakianETB12,MedyaSSBS18,ZhaoLTGX14,ZhaoWLTG17} and to \emph{Shapley} centrality~\cite{GomezGMOPT,MichalakASRJ13,SkibskiMR18,SzczepanskiMR16,TarkowskiMRW18Arxiv,TarkowskiSMHW18}.
In all these papers, the centrality indices are solely based on graph theoretical properties and do not take into account dynamic models for influence spread.

Modeling spread of influence and information diffusion has also been widely investigated in the literature. One of the most studied problems is the so-called \emph{influence maximization problem}: Given a network and a budget $k$, find a set of $k$ nodes, called \emph{seeds}, to be the starters of an influence diffusion process in such a way that the expected number of nodes that have been influenced at the end of the process is maximized. 
The influence maximization problem has been introduced by Domingos and Richardson~\cite{DomingosR01,RichardsonD02} and formalized as an optimization problem by Kempe et al.~\cite{DBLP:journals/toc/KempeKT15}. Several work followed these seminal papers, we refer the interest reader to~\cite{ChenCL13} and to the references in~\cite{DBLP:journals/toc/KempeKT15}. 
In the literature on influence maximization, the influence capability of a set of seeds is modeled as a function $\sigma$: given a set $S$ of seed nodes, $\sigma(S)$ is the expected number of eventually influenced nodes. The definition of $\sigma$ depends on the model used to represent the spread of influence in the network, the two most popular models being the \emph{Independent Cascade Model (ICM)} and \emph{Linear Threshold Model (LTM)}, see Section~\ref{subsec: Preliminaries on GSV for SI} for more details.

As previously mentioned, Chen and Teng initiated studying the interplay of network centrality and dynamic models for influence spread~\cite{chen2017interplay}. They introduced two centrality indices that are based on the most commonly used models for influence diffusion. 
The first index is called Single Node Influence (SNI) and measures the importance of a node $v$ by its capability of influencing other nodes when $v$ is the only seed node, i.e., the SNI of $v$ is $\sigma(\{v\})$. 
The second index, called Shapley centrality, is computed as the Shapley value of each node when the payoff function is $\sigma$. The authors presented an axiomatic characterization of the proposed centrality indices, that is they presented five axioms and showed that the Shapley centrality is the only index that satisfies all of them, while the SNI centrality is the only index that satisfies a set of three different axioms. This characterization captures the differences between these two indices: while SNI is suitable to model the centrality of a single node when it acts alone, the Shapley centrality characterizes the additional influence of a single node when acting in a group. 
We remark that our IGS centrality extends this latter concept by measuring the added value (in terms of influence) of a group of nodes when it operates within a larger group.
In the same paper, Chen and Teng proposed an efficient algorithm to approximately compute SNI and Shapley centralities and experimentally evaluated it on several real-world networks.

In a follow-up paper, Chen et al~\cite{ChenTZ18Arxiv}, presented a unified framework to extend classical graph-theoretical centrality indices to influence based ones and group centrality indices to their Shapley influence-based counterparts.
They follow an axiomatic approach, that is, they show that the derived influence-based centrality formulations are the unique centrality indices that conform with their corresponding graph-theoretical ones and satisfy the Bayesian axiom. 
They also provide scalable algorithms to compute influence-based centrality and Shapley centrality. 
We remark that also in this latter paper the aim is to evaluate the centrality of a single node, while the focus of our paper is on evaluating the centrality of a group of nodes.

\section{PRELIMINARIES}
In this section, we aim to connect cooperative game theory with influence maximization. Subsection~\ref{subsec: Preliminaries on GSV} introduces the concept of the Group Shapley value as (arguably most popular) special case of probabilistic generalized values in cooperative game theory. 
We then provide some observations on the Group Shapley value that will turn out essential when we turn to its computation later on. Subsection~\ref{subsec: Preliminaries on GSV for SI} recalls the basics needed from the influence maximization literature. Subsection~\ref{subsec: Preliminaries on GSC} explains how the Group Shapley value can be applied to the setting of influence maximization, we refer to it as Influence-based Group Shapley (IGS) centrality in this case.

Throughout the paper, we denote by $[n]$ the set $\{1,\ldots,n\}$, by $2^{S}$ the set of all subsets of $S$, and by \(\binom{S}{k}\) the set of subsets of $S$ of size $k$. 

\subsection{The Group Shapley Value} \label{subsec: Preliminaries on GSV}
In cooperative game theory~\cite{chalkiadakis2011computational, myerson2013game}, a game on $n\ge 2$ players 
is commonly formalized by a characteristic function $\tau:2^{[n]}\rightarrow \RR$ that assigns to every subset $S\subseteq [n]$ of players, also called a coalition, a value $\tau(S)$.
Marichal et al.~\cite{DBLP:journals/dam/MarichalKF07} introduced the concepts of \emph{probabilistic generalized values} and \emph{generalized semivalues} as ways of measuring the worth of a coalition $S\subseteq [n]$. Their notions generalize the more classical concepts of probabilistic values and semivalues from individuals to groups of individuals. That is, these values quantify the prospect of groups of players in a game. For fixed $n$, a probabilistic generalized value of a coalition $S\subseteq [n]$ in a game $\tau:2^{[n]}\rightarrow \RR$ is of the form
$
  \phi_\tau(S) := \sum_{T\subseteq [n]\setminus S} p^S_T \cdot (\tau(T \cup S) - \tau(T)),
$
where $p^S$ denotes a probability distribution on the subsets $T$ of $[n]\setminus S$.
That is, generally speaking, a probabilistic generalized value quantifies the average marginal contribution of the set $S$ to any set $T$ of players that is disjoint from $S$. Note that these marginal contributions are assigned different probabilities $p^S_T$. 
For fixed $S$, these probabilities can be understood as a-priori likelihoods of sets $T$ to be extended by $S$.
A generalized semivalue is a probabilistic generalized value such that $p^S_T=p^{S'}_{T'}$ if $|S|=|S'|$ and $|T|=|T'|$.

The arguably best known instance of generalized semivalues is the \emph{Group Shapley value}~\cite{DBLP:journals/4or/FloresMT19,DBLP:journals/dam/MarichalKF07}.
For a subset $S\subseteq [n]$ of players in a game $\tau$, the \emph{Group Shapley value} of $S$ is defined as
\[
  \shap_\tau(S):=\sum_{T\subseteq [n]\setminus S} \frac{|T|!(n-|S|-|T|)!}{(n-|S|+1)!} \cdot (\tau(T \cup S) - \tau(T)),
\]
i.e., it is the generalized semivalue for which 
$p^S_T=\frac{1}{n-|S|+1}/ \binom{n-|S|}{|T|}$. Stated otherwise, the set $T$ can be seen as a random variable chosen by first sampling an integer \(t\in\{0,\ldots,n-|S|\}\) uniformly at random and then picking a set of size $t$ in \([n]\setminus S\) uniformly at random.  

The Group Shapley value is a generalization of the well-known Shapley value~\cite{shapley1953value}, which quantifies the contribution of a single player to a coalition in a game. The standard \emph{Shapley value} of a player $i$ w.r.t.\ $\tau$ can simply be defined as the Group Shapley value of the singleton set $\{i\}$, that is  $\shap_\tau(i):=\shap_\tau(\{i\})$. It is well known that, using a standard counting argument, this definition is equivalent to
\(
  \shap_\tau(i):=\E_{\pi}[\tau(T_{\pi, i} \cup \{i\}) - \tau(T_{\pi, i})],
\)
where $\pi\sim \Pi([n])$ is a permutation of $[n]$ that is picked uniformly at random among all permutations $\Pi([n])$ of $[n]$ and $T_{\pi, i}$ denotes the set of players in $[n]$ ordered before $i$ in $\pi$.
A similar formulation in terms of permutations can also be obtained for the \emph{Group} Shapley value. Let us first introduce the following notation: For a set \(X\subseteq [n]\), let \(X_{\overline{S}} := (X\setminus S) \cup \{ \hat{s}\}\), where \(\hat{s}\) is an auxiliary item representing all the items from \(S\).\footnote{Flores et al.~\cite{DBLP:journals/4or/FloresMT19} define the so-called \emph{merging game} on $V_{\overline{S}}$. Considering games that result by merging players is quite common, see for example the work of Lehrer~\cite{lehrer1988axiomatization}.}
\begin{observation}[Group Shapley formulation using permutations] \label{obs: group shapley with permutation}
    The Group Shapley value of a set \(S\) in a game \(\tau\) is equal to
    \[
      \shap_\tau(S) = \E_{\pi \sim \Pi([n]_{\overline{S}})}[\tau(T_{\pi,\hat{s}}\cup S) - \tau(T_{\pi,\hat{s}})],
    \] 
    where \(T_{\pi,\hat{s}}\) is the subset of $[n]\setminus S$ preceding \(\hat{s}\) in the permutation \(\pi\) of $[n]_{\overline{S}}$ picked uniformly at random from $\Pi([n]_{\overline{S}})$.
\end{observation}
\begin{proof}
    Consider the game \(\hat{\tau}\) defined on \([n]_{\overline{S}}\) by \(\hat{\tau}(T) = \tau(T)\) if $\hat{s} \not\in T$ and \(\hat{\tau}(T) = \tau((T\setminus\{\hat{s}\}) \cup S )\) otherwise. Then Observation~\ref{obs: group shapley with permutation} follows from the fact that the Shapley value of $\hat{s}$ in $\hat{\tau}$ coincides with the Group Shapley value of $S$ in \(\tau\).    
\end{proof}

We proceed with the following observation on the probability that a given set $R$ intersects with a set $T$ chosen according to probabilities $p^S_T$ for the Group Shapley value. This observation relies on the Group Shapley formulation using the permutations.
\begin{observation}[Intersection Probability for the Group Shapley Value]\label{obs: intersection probability shapley}
  Let $S\subseteq [n]$ be a group. For any \(R\subseteq [n]\) with \(R\cap S \neq \emptyset\), it holds that \(\Pr_{\pi\sim \Pi([n]_{\overline{S}})}[R_{\overline{S}} \cap T_{\pi,\hat{s}} = \emptyset] = 1/(|R\setminus S| +1)\), where \(T_{\pi,\hat{s}}\) is the subset of $[n]$ preceding \(\hat{s}\) in the random permutation \(\pi\) of $[n]_{\overline{S}}$.
\end{observation}
\begin{proof}
  The event \(R_{\overline{S}}\cap T_{\pi,\hat{s}}=\emptyset\) is equivalent to the permutation \(\pi\) of $[n]_{\overline{S}}$ placing \(\hat{s}\) ahead of all the other nodes in \(R_{\overline{S}}\). Since \(\pi\) is sampled uniformly at random from \(\Pi([n]_{\overline{S}})\), this event happens with probability exactly \(1/|R_{\overline{S}}|=1/(|R\setminus S| +1)\).
\end{proof}

\subsection{Influence Maximization} \label{subsec: Preliminaries on GSV for SI}
We will be interested in Generalized Semivalues for functions that describe influence on information propagation in social networks. Two of the most popular models for describing information propagation in networks are the \emph{Independent Cascade} and \emph{Linear Threshold} models~\cite{DBLP:journals/toc/KempeKT15}. In both of these models, we are given a directed graph \(G=(V, E)\) where $V$ is a set of $n$ nodes, values \(\{p_{uv} \in [0,1]: (u,v) \in E\}\) and an initial node set \(A\subseteq V\) called \emph{seed nodes}. A spread of influence from the set \(A\) is then defined as a randomly generated sequence of node sets $(A_t)_{t\in \mathbb{N}}$, where \(A_0=A\) and \(A_{t-1}\subseteq A_{t}\). These sets represent active users, i.e., we say that a node \(v\) is \emph{active} at time step \(t\) if \(v\in A_t\). This sequence converges as soon as \(A_{t^*}=A_{t^*+1}\), for some time step \(t^*\ge 0\) called the time of quiescence. For a set \(A\), we use \(\sigma(A) = \E[|A_{t^*}|]\) to denote the expected number of nodes activated at the time of quiescence when running the process with seed nodes \(A\). In influence maximization, a common objective is to find a set $A$  maximizing \(\sigma(A)\) under a cardinality constraint.

\paragraph{The Independent Cascade Model.}
In the \emph{Independent Cascade} (IC) model, the values \(\{p_{uv} \in [0,1]: (u,v) \in E\}\) are probabilities. The sequence of node sets $(A_t)_{t\in \mathbb{N}}$, is randomly generated as follows.  If \(u\) is active at time step \(t\ge 0\) but was not active at time step \(t-1\), i.e., \(u\in A_t\setminus A_{t-1}\) (with \(A_{-1}=\emptyset\)), it tries to activate each of its neighbors $v$, independently, and succeeds with probability \(p_{uv}\). In case of success, \(w\) becomes active at time step \(t+1\), i.e., \(v\in A_{t+1}\). 

\paragraph{The Linear Threshold Model.} In the \emph{Linear Threshold} (LT) model, the values \(\{p_{uv} \in [0,1]: (u,v) \in E\}\) are weights such that for each node $v$, $\sum_{(u,v)\in E} p_{uv} \le 1$. The sequence of node sets $(A_t)_{t\in \mathbb{N}}$, is randomly generated as follows. At time step \(t+1\), every inactive node $v$ such that $\sum_{(u,v)\in E, u\in A_t} p_{uv} \ge \theta_v$ becomes active, i.e., \(v\in A_{t+1}\), where thresholds $\theta_v$ are chosen independently and uniformly at random from the interval $[0, 1]$. 

\paragraph{The Triggering Model.} The IC and LT models can be generalized to what is known as the \emph{Triggering Model}, see~\cite[Proofs of Theorem 4.5 and 4.6]{DBLP:journals/toc/KempeKT15}. 
For a node \(v\in V\), let \(N_v\) denote all in-neighbors of \(v\). In the Triggering model, every node independently picks a \emph{triggering set} \(T_v \subseteq N_v\) according to some distribution over subsets of its in-neighbors.
For a possible outcome \(X = (T_v)_{v\in V}\) of triggering sets for the nodes in \(V\), let \(G_X = (V,E')\) denote the sub-graph of \(G\) where \(E' = \{(u,v)|v \in V, u\in T_v \}\). Moreover, let \(\rho_X(A)\) be the set of nodes reachable from \(A\) in \(G_X\), then \(\sigma(A)=\E_X[|\rho_X(A)|]\). The IC model is obtained from the Triggering model if for each directed edge $(u,v)$, $u$ is added to $T_v$ with probability $p_{uv}$. Differently, the LT model is obtained if each node $v$ picks at most one of its in-neighbor to be in her triggering set, selecting a node $u$ with probability $p_{uv}$ and selecting no one with probability $1 - \sum_{u \in N_v} p_{uv}$.

Interestingly, the Triggering Model allows for using a concept commonly referred to as reverse reachable sets.

\paragraph{Reverse Reachable (RR) Sets.}
We describe the process of generating so-called \emph{Reverse Reachable (RR) sets}~\cite{DBLP:conf/soda/BorgsBCL14,DBLP:conf/sigmod/TangXS14}. A random RR set $R$ is generated as follows~\cite{chen2017interplay}.
(1) Set $R=\emptyset$.
(2) Uniformly at random select a root node $v\in V$ and add it to $R$.
(3) Until every node in $R$ has a triggering set: Pick a node $u$ from $R$ that does not have a triggering set, sample its triggering set $T_u$ and add it to $R$. A random RR set $R$ can be equivalently generated as all nodes that can reach a uniformly at random sampled root node $v$ in a random graph $G_X$ sampled as in the Triggering Model~\cite{DBLP:conf/soda/BorgsBCL14}. We get the following lemma for the marginal contribution of a set, the lemma generalizes Lemma 22 in~\cite{chen2017interplay} from marginal contribution of nodes to sets of nodes.
\begin{lemma}[Marginal Contribution]\label{lem: marginal contribution}
  Let \(R\) be a random RR set. For any \(T \subseteq V\) and \(S\subseteq V\setminus T\):
  \begin{align*}
      \sigma(T) &= n\cdot \Pr_R[R \cap T \neq \emptyset],\\
      \text{ and }\quad
      \sigma(T\cup S)-\sigma(T) &= n\cdot \Pr_R[R\cap S \neq \emptyset \land R\cap T=\emptyset].
  \end{align*}
\end{lemma}
\begin{proof}
  Let \(X\) be a random outcome profile in the triggering model and let \(\rho_X(S)\) denote the set of nodes reachable from \(S\) in $G_X$. Then
  \begin{align*}
      \sigma(T) &= \E_X[|\rho_X(T)|]
                = \E_X[\sum_{v\in V} \ones_{v\in \rho_X(T)}]\\
                &= n\cdot \E_X [\E_{v\sim V} [\ones_{v \in \rho_X(T)}]]
                = n\cdot \Pr_{X,v\sim V} [v \in\rho_X(T)].
  \end{align*}
  Recall that a random RR set is equivalently generated as all nodes that can reach a uniformly at random sampled root node $v$ in a random outcome graph $G_X$. Hence, the above event is equivalent to \(R\cap T \neq \emptyset\) and the first claim follows.
  Similarly, we have
  \begin{align*}
    \sigma(T\cup S) - \sigma(T) &= \E_X[|\rho_X(T\cup S) \setminus \rho_X(T)|]\\
    &= n \cdot \E_X[\E_{u\sim V}[ \ones_{u\in \rho_X(T\cup S) \setminus \rho_X(T)}]]\\
    &= n\cdot \Pr_{X,u\sim V} [u \in \rho_X(T\cup S) \setminus \rho_X(T)].
  \end{align*}
  By a similar argument, the event \(u \in \rho_X(T\cup S) \setminus \rho_X(S)\) is equivalent to the event \(R \cap S \neq \emptyset \land R\cap T = \emptyset\). This shows the second claim.
\end{proof}

\subsection{The Group Shapley Centrality}\label{subsec: Preliminaries on GSC}
Chen and Teng~\cite{chen2017interplay} consider the Shapley value of nodes w.r.t. the influence spread function $\sigma$ in a social network modeled by the Triggering Model. They use the resulting \emph{Shapley centrality} $\shap_\sigma(i)$ for $i$ being a node in the network as a measure of centrality of node $i$. They furthermore show that this centrality measure satisfies and is uniquely characterized by certain axioms, similar to the axioms characterizing the standard Shapley value. In this work, we consider the Group Shapley value w.r.t. $\sigma$.
For a set $S$ of nodes, we call the Group Shapley value w.r.t.\ $\sigma$ the \emph{Influence-based Group Shapley} (IGS) centrality of $S$, referred to as $\shap(S)$ omitting $\sigma$ as an index:
\begin{equation}
\shap(S)= \E_{\pi \sim \Pi(V_{\overline{S}})}[\sigma(T_{\pi,s}\cup S) - \sigma(T_{\pi,s})]. \label{eq:Shapley Centrality}
\end{equation}

One of the main contributions of Chen and Teng~\cite{chen2017interplay} is an algorithm that approximates the Shapley centrality of every node. The key lemma in their analysis is that for a node $v\in V$, it holds that $\shap_\sigma(v)=n\E_R[\ones_{v\in R}/|R|]$, where the expected value is over random RR sets generated as described above. For the IGS centrality $\shap(S)$ of a set $S$, we show the following analogous lemma.
\begin{lemma}[IGS centrality via RR sets]\label{lem: shapley value identity}
  Let \(S\subseteq V\), it holds that \(\shap(S) = n \cdot \E_R[\frac{\ones_{R\cap S\neq \emptyset}}{|R\setminus S| + 1}]\).
\end{lemma}
\begin{proof}
  For the IGS centrality, it holds that
  \begin{align*}
      \shap(S)
      &= \E_{\pi \sim \Pi(V_{\overline{S}})}[\sigma(T_{\pi,s}\cup S) - \sigma(T_{\pi,s})]\\
      &= \E_{\pi \sim \Pi(V_{\overline{S}})}[n\cdot \Pr_R(R\cap S \neq \emptyset \land R \cap T_{\pi,s} = \emptyset)]\\
      &=n\cdot \E_R[\E_{\pi \sim \Pi(V_{\overline{S}})}[\ones_{R\cap S \neq \emptyset \land R\cap T_{\pi,s} = \emptyset}]]\\
      &= n\cdot \E_R\Big[\frac{\ones_{R\cap S\neq \emptyset}}{|R\setminus S| + 1}\Big],
  \end{align*}
using Lemma~\ref{lem: marginal contribution} and Observation~\ref{obs: intersection probability shapley}. 
\end{proof}
Consequently from this lemma, we obtain that the IGS centrality is a monotonously increasing set function. Furthermore, Lemma~\ref{lem: shapley value identity} provides the following observation on the range of IGS centralities.

\begin{observation}[Range of IGS Centralities]\label{obs: range}
  Let $\Sshap$ be a set of size $k$ maximizing $\shap$. Then $\shap(\Sshap)\ge 1$. Moreover, $\shap(S)\ge \frac{k}{n}$ for any $S\subseteq V$ of size $k$. Lastly, $\shap(S)\le n$ for any $S$ and $\shap(V)=n$.
\end{observation}
\begin{proof}
  According to the normalization axiom of the Shapley value (for single items), we have $\sum_{i\in V}\shap(\{i\})=n$, hence there is a node $i_0$ for which $\shap(\{i_0\})\ge 1$.
  Since $\shap$ is a monotonously increasing set function, it holds that $\shap(S)\ge 1$ for any set $S$ of size $k$ containing $i_0$, thus also $\shap(\Sshap)\ge 1$. In order to show that $\shap(S)\ge\frac{k}{n}$ for any $S\subseteq V\setminus\emptyset$ of size $k$. With $R(u, X)$ we denote the RR set sampled from a node $u$ for an outcome profile $X$. We then observe that, according to Lemma~\ref{lem: shapley value identity}, $\shap(S)$ equals
  \[
    n \cdot \E_R\Big[\frac{\ones_{R\cap S\neq \emptyset}}{|R\setminus S| + 1}\Big] \ge n\cdot \frac{1}{n} \cdot \sum_{v\in S}\E_X\Big[\frac{\ones_{R(v,X)\cap S\neq \emptyset}}{|R(v,X)\setminus S| + 1}\Big]\ge \frac{k}{n},
  \]
  using that $|R(v,X)\setminus S|\le n-1$ as $v\in S$.
  The remaining claims follow analogously using the same equality from Lemma~\ref{lem: shapley value identity}.
\end{proof}
Our ultimate goal would be to find a set \(S\) of size at most $k$ with highest IGS centrality among all such sets. This is formalized below.
\begin{cproblem}{\MGSS}
Input: Influence maximization instance on digraph $G$, integer \(k\).

Find: \(S \subseteq V\) s.t.\ \(|S| \le k\), maximizing \(\shap(S)\).
\end{cproblem}
The naive approach for solving the optimization problem \MGSS would be to evaluate $\shap$ for all subsets of size at most $k$ and pick the one with highest value. Unfortunately, the formula for computing $\shap$ for a single set \(S\) given in Equation~\ref{eq:Shapley Centrality} is already not practical as it requires to compute the difference \(\sigma(S\cup T) - \sigma(T)\) for an exponential number of sets \(T\). Alternatively, one could try to follow an approach similar to the one taken by Chen and Teng~\cite{chen2017interplay} for the Shapley centrality of single nodes. Such approach for IGS centrality however would require updating $O(n^k)$ estimates (one for each candidate set) in every iteration. We will see later on in Section~\ref{sec: approximation algorithm} how to avoid this taking a different route. 
In the next section, we focus on overcoming the first difficulty, i.e., we show how to approximate IGS centrality. The approach relies on the representation of $\shap$ given in Lemma~\ref{lem: shapley value identity} and, non-surprisingly, on a Chernoff bound.

\section{EVALUATING IGS CENTRALITY}
This section is concerned with the question of estimating the function $\shap$. We first give a straightforward result that shows how to compute $\shap(S)$ for a given set $S$. Thereafter, we show that by sampling a sufficient number of RR sets, we can give a set function $\GSC$ that with high probability approximates $\shap$ in a sense that suffices for obtaining an approximation algorithm for \MGSS. The main tool for this section is the following classical Chernoff bound that can be found in the survey by Chung and Lu~\cite[Theorem 4]{chung2006concentration} or in Appendix C.2 of the full version of Chen and Teng's paper~\cite{chen2017interplay}.
\begin{fact}[Chernoff Bound]\label{fact: chernoff}
    Let $Y$ be the sum of $t$ i.i.d. random variables with mean $\mu$ and value range $[0,1]$.
    \begin{enumerate}
        \item\label{item 01} For any $\alpha\in (0,1)$, we have
        $
            \Pr[\frac{Y}{t} - \mu \le -\alpha \mu]
            \le \exp(-\frac{\alpha^2}{2}t\mu).
        $
        \item\label{item pos}
        For any $\alpha>0$, we have
        $
            \Pr[\frac{Y}{t} - \mu \ge \alpha \mu]
            \le \exp(-\frac{\alpha^2}{2+\frac{2}{3}\alpha}t\mu).
        $
    \end{enumerate}
\end{fact}
\paragraph{Approximately Evaluating $\shap(S)$.}
In this paragraph, we show that using the above Chernoff bound we can, in a straightforward way, obtain a $(1\pm\epsilon)$-approximation $\tGSC(S)$ of $\shap(S)$ for any set $S\subseteq V$ and $\epsilon \in (0,1)$ by sampling $\Theta(n^2\epsilon^{-2}\log n)$ RR sets.
\begin{lemma} \label{lem:evaluate}
  Let $S\subseteq V$ and $\epsilon\in (0,1)$. Let $R_1,\ldots, R_t$ be a sequence of $t\ge 6n^2\epsilon^{-2}c\log(n)$ RR sets for some constant $c\ge 2$. Then, with probability at least $1-n^{-c}$, it holds that
  $\tGSC(S) := \frac{n}{t} \sum_{i=1}^t \frac{ \ones_{R_i\cap S\neq \emptyset}}{|R_i\setminus S| + 1}$
  satisfies $|\tGSC(S) - \shap(S)| < \epsilon\cdot \shap(S)$.
\end{lemma}
\begin{proof}
   If $S=\emptyset$, the statement trivially holds. Otherwise, define the random variables $Y_i(S):=\frac{\ones_{R_i\cap S\neq \emptyset}}{|R_i\setminus S| + 1}\in[0,1]$ for $i\in[t]$ and let $Y(S):=\sum_{i=1}^t Y_i(S)$ as well as $\mu(S):=\E[Y_i(S)]$.
    Clearly, $\mu(S)=\shap(S)/n$ by Lemma~\ref{lem: shapley value identity} and $\tGSC(S)=\frac{n}{t}Y(S)$. Thus $\Pr[|\tGSC(S) - \shap(S)| \ge \epsilon \cdot \shap(S)]$ equals
    \begin{align*}
        \Pr\Big[\Big|\frac{Y(S)}{t} - \mu(S)\Big| \ge \epsilon \cdot \mu(S)\Big]
        &\le 2\exp\Big(-\frac{\epsilon^2}{3}\cdot t\mu(S)\Big)
    \end{align*}
    using Fact~\ref{fact: chernoff} and $\epsilon\le 1$. We lower bound $\mu(S)=\phi(S)/n\ge 1/n^2$ using Observation~\ref{obs: range}. The choice of $t$ leads to the bound of $n^{-c}$.
\end{proof}

\paragraph{An Approximate Characterization $\GSC$ of $\shap$.} 
Lemma~\ref{lem:evaluate} is unsatisfactory for the following reason. It samples a number of RR sets that is quadratic in the number of nodes $n$ even for evaluating the IGS centrality of a single group. In this paragraph, we show how to circumvent this problem. We show that a near-linear number of RR sets suffices to compute, for any set \(S\), an approximation \(\GSC(S)\) of $\shap(S)$ that is good enough for giving an approximation algorithm for \MGSS.
More precisely, we define a function \(\GSC\) that meets the following two conditions: (1) For any set \(S\) of size $k$, \(\GSC(S)\) does not overestimate \(\shap(S)\) too much and (2) For an optimal set $\Sshap$, \(\GSC(\Sshap)\) does not underestimate \(\shap(\Sshap)\) too much. We will show that these conditions suffice for a set \(S\) that is close to being optimal for \(\GSC\) to also be close to being optimal for \(\shap\). 

\begin{thrm} \label{thm:link shap hmhs}
  Let $\epsilon\in (0,1)$ and $R_1,\ldots, R_t$ be a sequence of RR sets of length $t\ge 6n\epsilon^{-2}(c+k)\log(n)$ for some constant $c\ge 2$.  Let $S^*$ be a set of size at most $k$ maximizing $\shap$ and let
  \[
    \GSC(S):=\frac{n}{t} \sum_{i=1}^t \frac{ \ones_{R_i\cap S\neq \emptyset}}{|R_i\setminus S| + 1} \quad \text{ for each set }S\in \binom{V}{k}.
  \]
  Then, with probability at least $1-n^{-c}$, the following conditions hold.
  \begin{enumerate}
      \item[(1)] $\GSC(S) - \shap(S)< \epsilon\cdot \shap(\Sshap)$, for all $S\in \binom{V}{k}$.
      \item[(2)] $\GSC(\Sshap) - \shap(\Sshap) > -\epsilon\cdot \shap(\Sshap)$.
  \end{enumerate}
\end{thrm}
\begin{proof}
    We first show that (1) holds with probability at least $1-\frac{1}{2}n^{-c}$. Let $S\in \binom{V}{k}$. Define the random variables $Y_i(S):=\frac{\ones_{R_i\cap S\neq \emptyset}}{|R_i\setminus S| + 1}\in[0,1]$ for every $i\in[t]$ and let $Y(S):=\sum_{i=1}^t Y_i(S)$ as well as $\mu(S):=\E[Y_i(S)]$.
    Clearly, $\mu(S)=\shap(S)/n$ by Lemma~\ref{lem: shapley value identity} and $\GSC(S)=\frac{n}{t}Y(S)$. Let $\alpha:=\frac{\epsilon \shap(\Sshap)}{\shap(S)}$. Then, $\Pr[\GSC(S) - \shap(S) \ge \epsilon \cdot \shap(\Sshap)]$ equals
    \begin{align*}
        \Pr\Big[\frac{Y(S)}{t} - \mu(S) \ge \alpha \cdot \frac{\shap(S)}{n}\Big]
        &\le \exp\Big(-\frac{\alpha^2}{2+\frac{2}{3}\alpha}\cdot t\mu(S)\Big),
    \end{align*}
    using Fact~\ref{fact: chernoff}. Using the definition of $\alpha$, we get that the argument of $\exp$ is equal to $-\frac{\epsilon^2 \shap(\Sshap)}{2 \shap(S)/\shap(\Sshap)+2\epsilon/3}\cdot \frac{t}{n}\le -\frac{\epsilon^2 \shap(\Sshap)}{3n}\cdot t$ using that $\epsilon<1$ and $\shap(S)\le \shap(\Sshap)$.
    Using that $\shap(\Sshap)\ge 1$ according to Observation~\ref{obs: range} and the definition of $t$ lead to the upper bound of $1/n^{2(c + k)}$.
    A union bound over all at most $n^k$ sets in $\binom{[n]}{k}$ shows that $\GSC(S) - \shap(S) \ge \epsilon \cdot \shap(\Sshap)$ holds for every such $S$ with probability at most $n^k\cdot n^{-2(c+k)}\le \frac{1}{2}\cdot n^{-c}$.

    We proceed to condition (2). It holds that $\Pr[\GSC(\Sshap) - \shap(\Sshap) \le -\epsilon \cdot \shap(\Sshap)]$ is equal to 
    \begin{align*}
        \Pr\Big[\frac{Y(\Sshap)}{t} - \mu(\Sshap) \le -\epsilon \cdot \frac{\shap(\Sshap)}{n}\Big]
        &\le \exp\Big(-\frac{\epsilon^2}{2}\cdot t\mu(\Sshap)\Big)
    \end{align*}
    using Fact~\ref{fact: chernoff} with $\alpha=\epsilon$. 
    Furthermore, $\mu(\Sshap)=\shap(\Sshap)/n$ and again $\shap(\Sshap)\ge 1$ as well as the definition of $t$ yield that $\GSC(\Sshap) - \shap(\Sshap) \le -\epsilon \cdot \shap(\Sshap)$ holds with probability at most $\frac{1}{2} \cdot n^{-c}$. A union bound over the probabilities that (1) or (2) do not hold, concludes the proof.
\end{proof}

In the next section, we investigate how to find an approximation algorithm for the \MGSS problem without computing the centralities of all sets of size $k$. 

\section{FINDING GROUPS OF LARGE IGS CENTRALITY}\label{sec: approximation algorithm}
To address the \MGSS problem, Theorem~\ref{thm:link shap hmhs} suggests the following approach. Sample a near-linear number $t$ of RR sets and compute a set of nodes $S$ that maximizes $\GSC(S)$. We formalize this problem as a variant of the well known \MHS problem, that we call the \HMHS problem.
\begin{cproblem}{\HMHS}
  Input: set \(X = \{x_1,\ldots,x_n\}\), set \(Z = \{Z_1, \ldots,Z_m\}\) of subsets of \(X\), integer \(k\).

  Find: $S\subseteq X$ s.t.\ $|S|\le k$ maximizing 
  \(    
    f_Z(S):=\sum_{i=1}^m \frac{\ones_{Z_i \cap S \neq \emptyset}}{|Z_i\setminus S|+1}.
  \)
\end{cproblem}
It is a non-linear variant of the well-known \MHS problem (which is itself equivalent to the \textsc{Max-Set-Cover}\xspace problem~\cite{garey2002computers}) in which the objective function is \(\sum_{i=1}^m \ones_{Z_i \cap S \neq \emptyset}\).

The problem of maximizing the previously defined function \(\GSC\) can be stated as a \HMHS problem by letting \(X\) be the set $V$ of nodes in graph $G$ and \(Z\) be the set of generated RR sets. The connection between the \HMHS and \MGSS problems from an approximation algorithm's perspective is made more formal in the next lemma.

\begin{lemma}\label{lem: MGSS HMHS}
    Let $\alpha\in(0,1]$, $\epsilon \in (0, 1)$, $c\ge 2$ and $k\in [n]$. Let $S_{\alpha}$ be an \(\alpha\)-approximate solution for the \(\HMHS\) problem with budget $k$, $X=V$ and $Z=\{R_1,\ldots, R_t\}$ s.t. $t\ge 24n\epsilon^{-2}(c+k)\log(n)$ RR sets. Then, $S_{\alpha}$ is an \((\alpha - \epsilon)\)-approximation for \MGSS with probability at least $1-n^{-c}$.
\end{lemma}
\begin{proof}
    Let  \(\epsilon' = \epsilon/2\). By Theorem~\ref{thm:link shap hmhs}, we have that $\GSC(S) - \shap(S)< \epsilon'\cdot \shap(\Sshap)$ for all $S\in \binom{[n]}{k}$ and $\GSC(\Sshap) - \shap(\Sshap) > -\epsilon'\cdot \shap(\Sshap)$ hold with probability at least $1-n^{-c}$,
    where $\GSC(S):=\sum_{i=1}^t \frac{n\cdot \ones_{R_i\cap S\neq \emptyset}}{t\cdot (|R_i\setminus S| + 1)}$.
    Let \(\Sshap\) (resp.\ \(\hat{S}^*\)) be a set of size $k$ maximizing \(\shap\) (resp.\ \(\GSC\)). Then with probability at least $1-n^{-c}$, it holds that
    \begin{align*}
        \shap(S_{\alpha}) &+ \epsilon' \cdot \shap(\Sshap)
        \ge \GSC(S_{\alpha})
        \ge \alpha \cdot \GSC(\hat{S}^*)\\
        &\ge \alpha \cdot \GSC(\Sshap)
        \ge \alpha \cdot (\shap(\Sshap) - \epsilon'\cdot \shap(\Sshap)),
    \end{align*}
    where we have used that \(\GSC(S_{\alpha})
        \ge \alpha \cdot \GSC(\hat{S}^*)\) as \(\GSC(\cdot) = \frac{n}{t} f_Z(\cdot)\). 
    The choice of $\epsilon'$ and $\alpha\le 1$ yield $\alpha - \alpha \epsilon'-\epsilon' \ge \alpha-\epsilon$. Thus $\shap(S_{\alpha}) \ge (\alpha - \epsilon)\cdot \shap(\Sshap)$ with probability at least $1-n^{-c}$.
\end{proof}

\paragraph{Approximation Algorithm.}
In this section, we describe a \(\frac{1-1/e}{k}\) approximation algorithm for the \HMHS problem. Consider an instance \((X,Z,k)\) and define the following set function
\(
    h_{Z}(S) := \sum_{i=1}^m  \ones_{Z_i \cap S \neq \emptyset}/|Z_i|.
\)
Note the similarity between $h_Z$ and $f_Z$. In fact, the approximation algorithm that we propose is to greedily maximize  \(h_{Z}\) instead of \(f_{Z}\). Why would this be a good idea? 
(1) The set function \(h_{Z}\) is monotone and submodular; thus the greedy algorithm will yield a \(1-1/e\) approximation to maximizing \(h_{Z}\). (2) Given a set \(S\subseteq X\) with \(|S|\le k\), it holds that
\begin{equation}
    f_{Z}(S) \ge h_{Z}(S) \ge f_{Z}(S)/k, \label{eq:encadrement}
\end{equation}
that is, the error when considering $h_Z$ instead of $f_Z$ is bounded by $k$. 
Hence, if we denote by \(S^{*}_f\) (resp. \(S^{*}_h\)) an optimal solution of size \(k\) for maximizing \(f_{Z}\) (resp. \(h_{Z}\)), we have that \(h_{Z}(S^{*}_h) \ge h_{Z}(S^{*}_f) \ge f_{Z}(S^{*}_f)/k\).
Now let \(S\) be the solution of size $k$ returned by the greedy algorithm. Then, 
\(S\) is a \(\frac{1-1/e}{k}\) approximation to maximizing \(f_{Z}\) as
\begin{align*}
    f_Z(S) &\ge h_Z(S) \ge \Big(1 - \frac{1}{e}\Big) \cdot h_{Z}(S^{*}_h)
    \ge \frac{1 - 1/e}{k} \cdot f_{Z}(S^{*}_f).
\end{align*}
It remains to prove the inequalities in~\eqref{eq:encadrement}.
\begin{proof}[Proof of Inequalities in~\eqref{eq:encadrement}]
    We note that, for any $S$, $f_{Z}(S)$ equals   
    \begin{align}\label{eq:fh}
        \sum_{i=1}^m  \frac{\ones_{Z_i \cap S \neq \emptyset}}{|Z_i\setminus S| + 1}
        = \sum_{i=1}^m  \frac{|Z_i|}{|Z_i|-|Z_i\cap S| + 1} \frac{\ones_{Z_i \cap S \neq \emptyset}}{|Z_i|}.
    \end{align}
    The left inequality in~\eqref{eq:encadrement} follows since $|Z_i|\ge |Z_i| - |Z_i\cap S| + 1$, if $Z_i \cap S \neq \emptyset$.
    Next, we observe that $\frac{|Z_i|}{|Z_i|-|Z_i\cap S| + 1}=1 + \frac{|Z_i\cap S| - 1}{|Z_i|-|Z_i\cap S| + 1}$ and, if $S$ is of size $k$, $|Z_i\cap S|\le k$. Together with $|Z_i|-|Z_i\cap S|\ge 0$, we obtain $\frac{|Z_i|}{|Z_i|-|Z_i\cap S| + 1} \le k$. The right inequality in~\eqref{eq:encadrement} follows.
\end{proof}
Using Lemma~\ref{lem: MGSS HMHS}, we thus obtain the following theorem.
\begin{thrm}\label{thrm: approximation algorithm}
  Let $\epsilon\in (0, 1)$ and $c\ge 2$. Using $\Theta(nk\epsilon^{-2}\log n)$ RR sets, we can obtain a $\frac{1-1/e}{k}-\epsilon$ approximation to the \MGSS problem with probability at least $1-n^{-c}$.
\end{thrm}
We conclude this paragraph with a note on an alternative approach for an approximation algorithm with a similar ratio but worse dependency on the number of generated RR sets. Define the function $h(S):= n \cdot \E_R[\frac{\ones_{R\cap S\neq \emptyset}}{|R|}]$. Note that $\shap(S)\ge h(S)\ge \shap(S)/k$ holds following a proof analogous to the one of~\eqref{eq:encadrement}. Clearly, $h$ is a monotone and submodular set function (just as $h_Z$ is) and it can be approximated within a $1\pm \epsilon$ factor for any $\epsilon \in (0,1)$ in an analogous way as used in Lemma~\ref{lem:evaluate}. Thus the greedy algorithm can be used in order to maximize $h$ to within a $1 - 1/e - \epsilon$ factor subject to the cardinality constraint. Altogether, we obtain the approximation ratio of $\frac{1 - 1/e - \epsilon}{k}$ which is comparable to what is achieved by Theorem~\ref{thrm: approximation algorithm}. However the required number of RR sets in every step of the $k$ steps of this greedy naive approach is quadratic in the number of nodes.

While Theorem~\ref{thrm: approximation algorithm} provides an interesting result for small \(k\) values, it remains unsatisfactory for large $k$. One could hope for a much stronger result as for example constant-factor approximations. Unfortunately, this is unlikely as we will see in the following section, where we provide several approximation hardness results for the \MGSS problem.

\section{HARDNESS OF APPROXIMATION}
In this section, we show that \MGSS under the IC model is, up to a constant factor, as hard to approximate as \DKS. More precisely, we prove the following theorem.
\begin{thrm} \label{thm: GS hardness}
  Let $\alpha\in(0,1]$. If there is an \(\alpha\)-approximation algorithm for \MGSS, then there is an \(\alpha/8\)-approximation algorithm for \DKS.
\end{thrm}
A number of strong hardness of approximation results are known for \DKS. We review some of them:
(1) \DKS cannot be approximated within $1/n^{o(1)}$ if the Gap Exponential Time Hypothesis (Gap-ETH) holds~\cite{manurangsi2017almost}. (2) \DKS cannot be approximated within any constant if the Unique Games with Small Set Expansion conjecture holds~\cite{raghavendra2010graph}. (3) \DKS cannot be approximated within $n^{-(\log\log n)^{-c}}$ for some constant $c$ if the Exponential Time Hypothesis holds~\cite{manurangsi2017almost}. 
Using the reduction given in this section, we obtain the same hardness results also for \MGSS. In particular, we would like to stress that, according to (1) and our reduction, it is unlikely to find anything better than an $(n^{-c})$-approximation for \MGSS, where $c$ is a constant. Furthermore, for all settings where $k= O(n^{c})$, such an algorithm is implied by our result in Section~\ref{sec: approximation algorithm}. We proceed by formally defining \DKS.
\begin{cproblem}{\DKS}
Input: Undirected graph \(G=(V, E)\), integer \(k\).

Find: set \(T \subseteq V\) with \(|T| \le k\), s.t.\ \(|E[T]|\) is maximum. Here $E[T]$ are the edges induced by $T$, i.e.\ $E[T]:=\{e \in E: e \subseteq T \}$.
\end{cproblem}
    
\paragraph{The reduction.} Let us fix an an instance \(\PPP=(G=(V,E),k)\) of \DKS. Note that, w.l.o.g., we can assume that \(G\) is connected.\footnote{It is not hard to show that from an \(\alpha\)-approximation algorithm for connected graphs, we can obtain an $(\alpha/2)$-approximation algorithm for general graphs by making the graph connected and then applying the approximation algorithm as follows.
Let us assume that $G$ is disconnected, and there are $\nu$ connected components, we construct a graph $\hat{G}=(V,\hat{E})$ by adding  $\nu-1$ edges to $E$ in order to make the graph connected.
Let $\hat{T}$ be a solution returned by an approximation algorithm for connected graphs on $\hat{G}$. If $\hat{T}$ contains both nodes at the endpoints of an edge in $\hat{E}\setminus E$, we can compute a solution $T$ for $G$ by iteratively substituting each pair of nodes in $\hat{T}$ that do not induce an edge in $G$ with two nodes that are adjacent in $G$.
It is easy to see that $|E[T]|$ is at least $\frac{1}{2}|\hat{E}[\hat{T}]|$, where $\hat{E}[\hat{T}]$ are the edges induced by $\hat{T}$ in $\hat{G}$: for each edge $e=(u,v)$ in $\hat{E}[\hat{T}]\setminus E$ either there exists an edge in $E[T]$ that is incident to $u$ or $v$, or $u$ and $v$ are substituted with two nodes that are adjacent in $G$.
Moreover, $|\hat{E}[\hat{T}]|\geq\alpha |\hat{E}[\hat{T}^*]|\geq\alpha |E[T^*]|$, where $\hat{T}^*$ and $T^*$ are optimal solution for $G$ and $\hat{G}$, respectively.
}
From \(\PPP\), we create the following \MGSS instance \(\overline \PPP=(\overline{G}=(\overline{V},\overline{A}),\{p_a\}_{a\in \overline{A}}, \overline{k})\): (1)~Probabilities \(p_a\) are set to 1 for all $a\in \overline{A}$. (2)~We set the budget to \(\overline{k} := k\cdot t\), where $t := 6|E|$. (3)~The node set is defined as \( \overline{V} := \overline{V}_V \cup \overline{V}_E\), where 
    $\overline{V}_V := \{u^v_1,\ldots,u^v_t|v \in V\}$
 and 
    $\overline{V}_E := \{u^e_1,\ldots,u^e_{\ell}|e\in E\}$ for $\ell := (2t+1)t|V|+1$. 
(4)~The arc set \(\overline{A}\) is defined as follows: for each edge \(\{v,v'\}\in E\), we create a pattern in \(\overline{G}\) as described in Figure~\ref{fig:edge}. This pattern is composed of three layers. Two layers, called $v$-layer and $v'$-layer, gathering all the nodes in $\{u^v_1,\ldots,u^v_t\}$ and $\{u^{v'}_1,\ldots,u^{v'}_t\}$ respectively and one layer called, $\{v,v'\}$-layer, gathering all nodes in $\{u^{\{v,v'\}}_1,\ldots,u^{\{v,v'\}}_{\ell}\}$. There is an arc from each node of the first two layers to each node of the third layer.

\begin{figure}[h!]
  \centering{
\scalebox{1}{
\begin{tikzpicture}[scale=.7,->,>=stealth',shorten >=1pt,auto,semithick]
        \node [draw] (V1)    at (0, 4)  {$u^v_1$};
        \node [draw] (V2)    at (1.5, 4) {$u^v_2$};
        \node (Dots1)        at (3, 4)  {$\cdots$};
        \node [draw] (VT)    at (4.5, 4)  {$u^v_t$};
        \node  (empty)    at (8, 4)  {$v$-layer};
        
        \node [draw] (VP1)   at (0, 0)  {$u^{v'}_1$};
        \node [draw] (VP2)   at (1.5, 0) {$u^{v'}_2$};
        \node (Dots3)        at (3, 0)  {$\cdots$};
        \node [draw] (VPT)   at (4.5, 0)  {$u^{v'}_t$};
        \node  (empty)    at (8, 0)  {$v'$-layer};
        
        \node [draw] (E1)    at (-3, 2)    {$u^{\{v,v'\}}_1$};
        \node [draw] (E2)    at (0, 2)  {$u^{\{v,v'\}}_2$};
        \node (Dots2)        at (2.25, 2) {$\cdots$};
        \node [draw] (ELM1)  at (4.5, 2) {$u^{\{v,v'\}}_{\ell-1}$};
        \node [draw] (EL)    at (7.5, 2) {$u^{\{v,v'\}}_\ell$};
        \node  (empty)    at (10.7, 2)  {$\{v,v'\}$-layer};

        \path [draw = black, rounded corners, inner sep=100pt,dotted]  
               (-1, 4.7) 
            -- (10.5,4.7)
	        -- (10.5,3.3)	 
            -- (-1,3.3)
            -- cycle;
        
        \path [draw = black, rounded corners, inner sep=100pt,dotted]  
               (-1, 0.7) 
            -- (10.5,0.7)
	        -- (10.5,-.7)	 
            -- (-1,-.7)
            -- cycle;
            
        \path [draw = black, rounded corners, inner sep=100pt,dotted]  
               (-4.5, 2.8) 
            -- (12.7,2.8)
	        -- (12.7,1.2)	 
            -- (-4.5,1.2)
            -- cycle;
        
        \path (V1) edge [] node {} (E1)
              (V1) edge [] node {} (E2)
              (V1) edge [] node {} (ELM1)
              (V1) edge [] node {} (EL)
              
              (V2) edge [] node {} (E1)
              (V2) edge [] node {} (E2)
              (V2) edge [] node {} (ELM1)
              (V2) edge [] node {} (EL)

              (VT) edge [] node {} (E1)
              (VT) edge [] node {} (E2)
              (VT) edge [] node {} (ELM1)
              (VT) edge [] node {} (EL)
              
              (VP1) edge [] node {} (E1)
              (VP1) edge [] node {} (E2)
              (VP1) edge [] node {} (ELM1)
              (VP1) edge [] node {} (EL)
              
              (VP2) edge [] node {} (E1)
              (VP2) edge [] node {} (E2)
              (VP2) edge [] node {} (ELM1)
              (VP2) edge [] node {} (EL)

              (VPT) edge [] node {} (E1)
              (VPT) edge [] node {} (E2)
              (VPT) edge [] node {} (ELM1)
              (VPT) edge [] node {} (EL);
      \end{tikzpicture}
  }
}
\caption{Pattern obtained in \(\overline{G}\) for each edge \(\{v,v'\} \in E\).}
\label{fig:edge}
\end{figure}
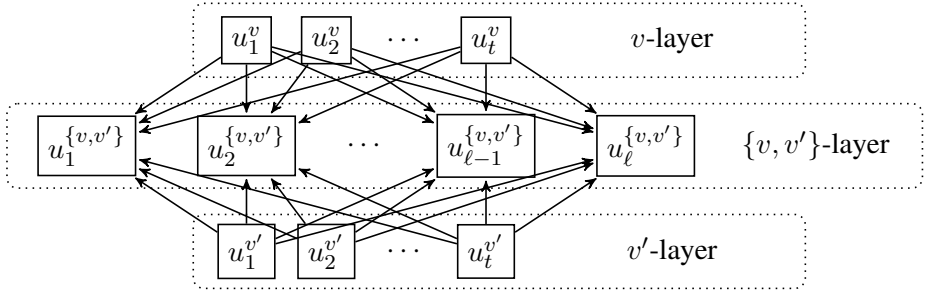

\paragraph{RR sets for $\overline \PPP$.}
Interestingly, all RR sets that can be generated in \(\overline \PPP\) are easily described as two different types:
(1) RR-sets that consist of singletons \(\{u^v_p\}\), we call them \emph{Node-RR-sets} and denote the set of all Node-RR-sets by \(\RRR_V\). There is exactly one Node-RR-set per node \(u^v_p \in \overline{V}_V\).
(2) RR-sets of the form \(\{u^{\{v,v'\}}_p, u^v_1,\ldots,u^v_t,u^{v'}_1,\ldots,u^{v'}_t\}\), we call them \emph{Edge-RR-sets} and denote the set of all Edge-RR-sets by \(\RRR_E\). There is exactly one Edge-RR-set per node \(u^{\{v,v'\}}_p \in \overline{V}_E\).
We note that each RR-set occurs with the same probability \(1/\overline{V}\). Hence, the IGS centrality of a set \(S\subseteq \overline{V}\), can be written as \(\shap(S) = \shap_V(S) + \shap_E(S)\), where \( \shap_V(S) := \sum_{R \in \RRR_V} \frac{\ones_{R \cap S \neq \emptyset}}{|R\setminus S| + 1}\) and \(\shap_E(S) := \sum_{R \in \RRR_E} \frac{\ones_{R \cap S \neq \emptyset}}{|R\setminus S| + 1} \). 

We proceed with a simple observation. For a node \(u \in \overline{V}\) and a set \(S \subseteq \overline{V}\setminus\{u\}\), we denote by \(I(u,S) := \shap(S \cup \{u\}) - \shap(S)\) the increase in IGS centrality obtained from adding $u$ to $S$. 
\begin{observation}
    Let \(u\in\overline{V}\) and \(S \subseteq \overline{V}\setminus\{u\}\). If \(u \in \overline{V}_V\), then \(I(u,S) \ge 1\). 
    If \(u \in \overline{V}_E\), then \(I(u,S) \le 1/2\). 
\end{observation} 
The proof of the above observation is simple. The first part holds due to the contribution of the Node-RR-set corresponding to \(u\) itself, the second part holds as \(u\) is in only one Edge-RR-set $R$ of size $2t+1$, thus 
$\frac{\ones_{R \cap (S\cup\{u\}) \neq \emptyset}}{|R\setminus (S\cup\{u\})| + 1} - \frac{\ones_{R \cap S \neq \emptyset}}{|R\setminus S| + 1} \le \frac{1}{2}$. 
Hence, we can say that a ``reasonable'' solution for \(\overline \PPP\) only contains nodes from \(\overline{V}_V\). 

\paragraph{Thorough Sets.}
We introduce some notation. For \(S\subseteq \overline{V}_V\) and \(v \in V\), we define \(\nb{S}{v}:=|S\cap\{u_1^v,\ldots,u_t^v\}|\), \(U_S:=\{v \in V : \nb{S}{v}\ge 1\}\), and \(U_S^{(1,t)} := \{v\in U_S : \nb{S}{v} \in (1,t) \}\). Note that $U_S, U_S^{(1,t)} \subseteq V$. 
The following notion is central to our analysis.
\begin{definition}
    We call a set $S\subseteq \overline V_V$ to be \emph{thorough}, if (1) \(\nb{S''}{v}=t\) for all $v\in U_{S}$ and (2) \(|E[U_{S}]|\ge 1\).
\end{definition}

Now, let \(S^*\) and \(T^*\) be optimal solutions for $\overline\PPP$ and $\PPP$, respectively. 
We get the following lemma, part (3) of which shows how to transform a thorough set $S$ into a solution of $\PPP$ in a straightforward way.
\begin{lemma} \label{lemma: hardness inequalities} 
    It holds that 
    (1)~\(\shap(S^*) \ge \frac{\ell}{2} |E[T^*]|\), 
    (2)~for all \( S\subseteq \overline{V}_V\), we have \(\shap_E(S) \ge \shap(S)/2\), and 
    (3)~if \( S\subseteq \overline{V}_V\) is thorough, we have \(\shap_E(S)\le \ell|E[U_S]|\).
\end{lemma}
\begin{proof}
    For~(1), we let \(S_{T^*} := \bigcup_{v\in T^*} \{u^v_i\}_{i\in [t]}\) and get \(\shap(S^*) \ge \shap(S_{T^*}) \ge \shap_E(S_{T^*})\). As each edge in $E[T^*]$ induces \(l\) Edge-RR-sets $R$ with $|R \setminus S_{T^*}| = 1$, we have \(\shap_E(S_{T^*}) \ge \frac{\ell}{2} |E[T^*]|\). 

    For~(2), fix \(S \subseteq \overline{V}_V\). As $G$ is connected, for any \(u\in \overline{V}_V\), there are at least $l$ Edge-RR-sets containing $u$. Hence, \(\shap_E(S)\ge \frac{\ell}{2t+1}\). On the other hand, note that \(\shap_V(S)\le t|V|\). As \(\ell > (2t+1)t|V|\), we get \(\shap_E(S) \ge \shap_V(S)\) and hence \(\shap_E(S) \ge \frac{\shap(S)}{2}\).

    For~(3), let \(S\subseteq \overline{V}_V\) be a thorough set. Define \(E_i := \{e\in E: e\cap U_S = i\}\) for $i=1,2$. Clearly, \(E[U_S]=E_2\) and, for each edge in \(E_1\) (resp.\ \(E_2\)) there are exactly \(\ell\) Edge-RR-sets \(R\) with \(|R\setminus S| = t + 1\) (resp.\ \(|R\setminus S| = 1\)). Hence, we conclude that
    \(
        \shap_E(S) 
        = \frac{\ell|E_2|}{2} + \frac{\ell|E_1|}{t+2}
        \le \ell |E_2|
    \)
    by the choice of $t$.
\end{proof}

In Lemmata~\ref{lemma: key hardness 1} and~\ref{lemma: key hardness 2} below, we show that every solution \(S\subseteq \overline{V}_V\) of $\overline \PPP$ can be transformed in polynomial time into a feasible thorough set $S''$ with \(\shap_E(S'') \ge \shap_E(S)\). This allows us to prove Theorem~\ref{thm: GS hardness}.
\begin{proof}[Proof of Theorem~\ref{thm: GS hardness}]
    Let \(S\) be an $\alpha$ approximate solution for $\overline \PPP$. We can assume, w.l.o.g., that \(S\subseteq \overline{V}_V\). Using (2) and (1) of Lemma \ref{lemma: hardness inequalities}, we get
    \(
       \shap_E(S) \ge \frac{\shap(S)}{2} \ge \frac{\alpha}{2} \shap(S^*) \ge \frac{\alpha \ell}{4}  |E[T^*]|.
    \)
    We now apply Lemmata~\ref{lemma: key hardness 1} and~\ref{lemma: key hardness 2} to the set \(S\), obtaining a thorough set \(S''\). Together with Lemma~\ref{lemma: hardness inequalities}~(3), we get
    $
       \shap_E(S) \le \shap_E(S'') \le \ell |E[U_{S''}]|.  
    $
    Thus, $S''$ is an $\frac{\alpha}{4}$-approximaion for $\PPP$. If $G$ is disconnected, this adds an extra $1/2$ to the approximation ratio. This concludes the proof.
\end{proof}

\paragraph{Transforming \(S\subseteq \overline{V}_V\) into a Thorough Set.}
Let \(S\subseteq \overline{V}_V\) be a set of size $\overline{k}$. We transform $S$ into a thorough set in two steps, the first of which is the following \emph{iterative process} that computes a set \(S'\) with \(|S'| = |S|\) by constructing a sequence \(S_0, \ldots, S_\mu\) with \(S_0=S\) and \(S_\mu=S'\). For a node $v\in V$, let \(I_E(v,S) := \shap_E(S_v) - \shap_E(S)\) where $S_v$ is the set obtained from $S$ by increasing \(\nb{S}{v}\) by one.  
While \(U_{S_i}^{(1,t)}\) contains at least two nodes \(v_h\) and \(v_l\) with \(I_E(v_h,S_i) \ge I_E(v_l,S_i)\), 
obtain \(S_{i+1}\) from \(S_i\) by increasing (resp.\ decreasing) \(\nb{S_i}{v_h}\) (resp.\ \(\nb{S_i}{v_l}\)) by one until one of the two nodes is not in \(U_{S_i}^{(1,t)}\). 
At the end, \(|U_{S_i}^{(1,t)}|\le 1\) and, moreover, the process terminates in polynomial time, since after at most $t$ iterations, one node is removed from \(U_{S_i}^{(1,t)}\).
\begin{lemma} \label{lemma: key hardness 1}
    The set \(S'\) satisfies \(\shap_E(S') \ge \shap_E(S)\).
\end{lemma}
\begin{proof}
    Let \(\Delta_i:=\shap_E(S_{i+1}) - \shap_E(S_{i})\) for every $i$ and $d(x) := \frac{1}{x(x+1)}$ and note that $d$ is decreasing in $x$. We show that \(I_E(v_h,S_{i})\ge I_E(v_l,S_i)\) implies \(\Delta_i\ge 0\) and \(I_E(v_h,S_{i+1})\ge I_E(v_l,S_{i+1})\).
    For a node \(v\in V\), let \(\RRR(v) := \{R\in \RRR_E : R\cap\{u^v_1,\ldots,u_t^v\}\neq \emptyset\}\) be all Edge-RR-sets that contain the nodes from $\overline{V}_V$ corresponding to $v$. 
    Then, we can rewrite \(I_E(v,S_i)\) and \(\Delta_i\) as
    \begin{align*}
      I_E(v,S_i) &= \sum_{R \in \RRR(v)} d(|R \setminus S_i|),\\    
      \Delta_i &= \!\!\!\!\!\!\!\!\sum_{R \in \RRR(v_h) \setminus \RRR(v_l)} \!\!\!\!\!\!\!\! d(|R\setminus S_{i}|)- \!\!\!\!\!\!\!\!\sum_{R \in \RRR(v_l)\setminus\RRR(v_h)} \!\!\!\!\!\!\!\! d(|R\setminus S_{i}|+1),
    \end{align*}
    As $d$ is decreasing,
    we have that 
    \begin{align*}
      \Delta_i
      \ge\!\!\!\!\!\!\!\!\sum_{R \in \RRR(v_h) \setminus \RRR(v_l)} \!\!\!\!\!\!\!\! d(|R\setminus S_{i}|)- \!\!\!\!\!\!\!\!\sum_{R \in \RRR(v_l)\setminus\RRR(v_h)} \!\!\!\!\!\!\!\! d(|R\setminus S_{i}|)
    \end{align*}
    which equals $I_E(v_h,S_i) - I_E(v_l,S_i)$. The latter is non-negative by choice of $v_h,v_l$.
    We turn to showing 
    \(I_E(v_h,S_{i+1})\ge I_E(v_l,S_{i+1})\). We have that $I_E(v_h,S_{i+1})=\sum_{R \in \RRR(v_h)}d(|R \setminus S_{i+1}|)$ equals
    \begin{align*}
        \!\!\!\!\!\!\!\!\sum_{R \in \RRR(v_h)\setminus \RRR(v_l)} \!\!\!\!\!\!\!\!\!d(|R \setminus S_{i}| - 1) + \!\!\!\!\!\!\!\!\!\sum_{R \in \RRR(v_h)\cap\RRR(v_l)} \!\!\!\!\!d(|R \setminus S_{i}|),
    \end{align*}
    which is at least $\sum_{R \in \RRR(v_h)} d(|R \setminus S_{i}|)=I_E(v_h,S_{i})$ as $d$ is decreasing. Similarly, $I_E(v_l,S_{i+1}) = \sum_{R \in \RRR(v_l)} d(|R \setminus S_{i+1}|)$ equals
    \begin{align*}
        \!\!\!\!\!\!\!\!\sum_{R \in \RRR(v_l)\setminus \RRR(v_h)} \!\!\!\!\!\!\!\!\!d(|R \setminus S_{i}| + 1) + \!\!\!\!\!\!\!\!\!\sum_{R \in \RRR(v_l)\cap \RRR(v_h)}\!\!\!\!\! d(|R \setminus S_{i}|),
    \end{align*}
    which is at most $\sum_{R \in \RRR(v_l)} d(|R \setminus S_{i}|) = I_E(v_l,S_{i})$.
    Hence, \(I_E(v_h,S_{i+1})\ge I_E(v_h,S_{i}) \ge I_E(v_h,S_{i}) \ge  I_E(v_l,S_{i+1})\).
\end{proof}

\begin{lemma}\label{lemma: key hardness 2}
    The set \(S'\) can be transformed into a thorough set $S''$ with \(|S''| \le |S'|\) and \(\shap_E(S'') \ge \shap_E(S')\).
\end{lemma}
\begin{proof}
    We start by treating a few trivial cases. (i) The set $S'$ is thorough. Then, we set $S'':=S'$. (ii)~Property (1) holds for $S'$, but $|E[U_{S'}]| = 0$ and (iii) \(\nb{S'}{v} < t\) for all nodes $v \in U_{S'}$. Recall that \(|U_{S'}^{(1,t)}| \le 1\) by construction of $S'$. Hence, (iii) implies that for all nodes $v\in U_{S'}$ but one, \(\nb{S'}{v} = 1\). In both cases (ii) and (iii), \(\shap_E(S')\le\ell|E|/(t+2)\), as \(|R\setminus S'| \ge t+1\) for each of the \(\ell|E|\) Edge-RR-sets $R$. Thus, by choosing any edge $e=\{v,v'\}$ and setting \(S'' := \{u^v_i,u^{v'}_i : i\in[t] \}\) we obtain a thorough set with $|S''|\le |S'|$ and \(\shap_E(S'') \ge \ell/2 \ge \ell|E|/(t+2) \ge \shap_E(S')\).

    If none of (i)-(iii) hold, we can order the nodes in \(U_{S'}\) such that there exists an index $r\in[|U_{S'}|]$ such that \(\nb{S'}{v_i} = t\) for $i\in[1, r-1]$, $\nb{S'}{v_{r}}\in [1, t)$, and \(\nb{S'}{v_{i}} = 1\) for $i\in [r+1, |U_{S'}|]$. Recall that $\overline{k}=kt$, thus \(\sum_{i=r}^{|U_{S'}|} \nb{S'}{v_i}\) is a (non-trivial) multiple of $t$. We conclude by distinguishing two more cases. If \(v_r\) is adjacent to one of \(\{v_i\}_{i\in[r-1]}\), we construct $S''$ from \(S'\) by setting $\nb{S'}{v_r} = t$ and $\nb{S'}{v_i} = 0$ for all $i\in [r+1, |U_{S'}|]$. Otherwise, we find an adjacent node \(v_q \in V\setminus\{v_i\}_{i\in[r-1]}\) (by connectivity of $G$ such node exists) and construct $S''$ from \(S'\) by setting $\nb{S'}{v_q} = t$ and $\nb{S'}{v_i} = 0$ for all \(i \in [r, |U_{S'}|]\).

    In both cases $S''$ is thorough and $|S''|\le |S'|$. It remains to show \(\shap_E(S'')\ge \shap_E(S')\).
    Indeed, in the first (resp.\ second) case \(v_r\) (resp.\ $v_q$) is in the neighborhood of \(\{v_i\}_{i\in[r-1]}\). By setting $\nb{S'}{v_r}$ (resp.\ $\nb{S'}{v_q}$) to $t$, there exist \(\ell\) Edge-RR-sets $R$ such that $|R\setminus S''| = 1/2$. The total increase of \(\shap_E\) on these RR-sets is at least \(\ell/6 = \ell/2 - \ell/3\). Conversely, the decrease resulting from setting $\nb{S'}{v_i} = 0$ for all \(i \in [r + 1, |U_{S'}|]\) (resp.\ for all \(i \in [r, |U_{S'}|]\)) is at most \(\ell|E|/(t+1)\), since for each of the Edge-RR-sets $R$ that intersect $\{v_{i}\}_{i\in [r + 1, |U_{S'}|]}$ (resp.\ $\{v_{i}\}_{i\in [r, |U_{S'}|]}$), it holds that $|R\setminus S'|\ge t$. As \(\ell/6 \ge \ell|E|/(t+1)\), this concludes the proof.
\end{proof}

\section{CONCLUSION AND FUTURE WORK}
We have formalized the problem of determining a set of $k$ nodes in a social network maximizing an influence-based Group Shapley centrality measure. Assuming common computational complexity conjectures, we have obtained strong hardness of approximation results for the problem at hand in this paper. For instance, this problem cannot be approximated within $1/n^{o(1)}$ under the Gap Exponential Time Hypothesis. On the other hand, we showed that a greedy algorithm achieves a factor of $\frac{1-1/e}{k}-\epsilon$ for any $\epsilon>0$, yielding an interesting result when $k$ is small.

Several directions for future work are conceivable. First, it would be worth investigating an algorithm with an approximation ratio which is sublinear in the number of nodes of the social network. Second, specific properties of the social network could allow more positive approximation results, as, for instance, the connectivity of the graph has a direct impact on the size of the generated reverse reachable sets. Hence, restricting this parameter could have an impact on the complexity of the problem from an approximation viewpoint. Third, it would be interesting to adapt our work to other generalized semivalues as, for instance, the Group Banzhaf value~\cite{DBLP:journals/dam/MarichalKF07}. Lastly, properly engineering and testing the approximation algorithm designed in this paper would be an interesting and complementary work.

\bibliographystyle{alpha}
\bibliography{biblio}

\end{document}